\documentclass[10pt]{IEEEtran}
\IEEEoverridecommandlockouts                            

\usepackage{color} 
\usepackage{caption}
\usepackage{amsfonts}
\usepackage{enumerate}
\usepackage{ragged2e}
\let\labelindent\relax
\usepackage{enumitem} 
\usepackage{etoolbox}
\usepackage{graphicx} 
\usepackage{latexsym,amsmath} 
\usepackage{mathrsfs} 
\usepackage{booktabs}
\usepackage{multicol}

\usepackage{nomencl}
\makenomenclature
\usepackage[makeroom]{cancel}
\usepackage{soul,xcolor}

\usepackage{amsthm}
\usepackage[font=small,labelfont=bf]{caption}

\usepackage[subtle]{savetrees} 

\title{Improving frequency response with synthetic damping available from fleets of distributed energy resources}

\author{Hani Mavalizadeh
    $\;\;\;$ Luis A. Duffaut Espinosa 
    $\;\;\;$ Mads R. Almassalkhi 
   \thanks{The authors are with the Department of Electrical and Biomedical Engineering, University of Vermont, Burlington, VT 05405 USA (Email: \{hani.mavalizadeh, lduffaut,malmassa\}@uvm.edu). This work was supported by the U.S. Department of Energy's ARPA-E award DE-AR0000694. M. R. Almassalkhi is co-founder of startup company Packetized Energy, which has commercialized aspects related to PEM.}
    
}

\allowdisplaybreaks[4]

\graphicspath{{figures/}}

\newtheorem{theorem}{Theorem}

\newtheorem*{remark}{Remark}
\newtheorem{definition}{Definition}

\theoremstyle{definition}

\begin{document}
    \bstctlcite{IEEEexample:BSTcontrol}
	\maketitle
\begin{abstract}
With the increasing use of renewable generation in power systems, responsive resources will be necessary to support primary frequency control in future low-inertia/under-damped power systems. Flexible loads can provide fast-frequency response services if coordinated effectively. However, practical implementations of such synthetic damping services require both effective local sensing and control at the device level and an ability to accurately estimate online and predict the available synthetic damping from a fleet. In addition, the inherent trade-off between a fleet being available for fast frequency response while providing other ancillary services needs to be characterized. In this context, the manuscript presents a novel, fully decentralized, packet-based controller for diverse flexible loads that dynamically prioritizes and interrupts loads to engender synthetic damping suitable for primary frequency control. Moreover, the packet-based control methodology is shown to accurately characterize the available synthetic damping in real-time, which is useful to aggregators and system operators.  Furthermore, spectral analysis of historical frequency regulation data is used to produce a probabilistic bound on the expected available synthetic damping for primary frequency control from a fleet and the trade-off from concurrently providing secondary frequency control services. Finally, numerical simulation on IEEE test networks demonstrates the effectiveness of the proposed methodology.
	\end{abstract}
	
	
\renewcommand\nomgroup[1]{%
  \item[\bfseries
  \ifstrequal{#1}{S}{Sets}{%
  \ifstrequal{#1}{I}{Indices}{%
  \ifstrequal{#1}{V}{Variables}{%
  \ifstrequal{#1}{P}{Parameters}{}}}}%
]}

\mbox{}

\nomenclature[S]{\(G\)}{Graph representing the topology of a grid}
\nomenclature[S]{\(\mathcal{V}\)}{Set of system busses}
\nomenclature[S]{\(\mathcal{E}\)}{Set of transmission lines}
\nomenclature[I]{\(i,j\)}{index for busses}
\nomenclature[P]{\(m_j\)}{Inertia constant of bus $j$}
\nomenclature[P]{\(D_j\)}{Damping coefficient in bus $j$ (MW/Hz)}
\nomenclature[P]{\(\Delta P^{G}_j\)}{Change in generation in bus j (MW)}
\nomenclature[P]{\(\Delta P^{L}_j\)}{Change in uncontrollable load in bus j (MW)}
\nomenclature[V]{\(\Delta P^{DER}_j\)}{Change in DER aggregate load in bus j (MW)}
\nomenclature[V]{\(\Delta P_{i,j}\)}{Power transfer between bus i and j (MW)}
\nomenclature[I]{\(k\)}{Index for time steps}
\nomenclature[V]{\(t_n[k]\)}{Local timer for $n^{th}$ DER at time k (sec)}
\nomenclature[I]{\(n\)}{Index for DERs}
\nomenclature[I]{\(h\)}{Index for harmonics}
\nomenclature[P]{\(\delta\)}{Packet length (sec)}
\nomenclature[P]{\(n_p\)}{Number of timer bins}
\nomenclature[V]{\(C_n\)}{A binary variable which is 1 if DER is ON, and 0 otherwise}
\nomenclature[P]{\(M\)}{A constant matrix used to model the timer dynamics}
\nomenclature[P]{\(B\)}{Vector allocating accepted DERs to timer's first bin}
\nomenclature[V]{\(q^{\text{ch}}\)}{Number of accepted charging requests}
\nomenclature[V]{\(x[k]\)}{The vector of number of DERs in each bin}
\nomenclature[P]{\(P^{\text{cap}}\)}{Power capacity of each DER (kW)}
\nomenclature[P]{\(T_D\)}{Time delay for the designed low pass filter (sec)}
\nomenclature[V]{\(\eta\)}{Calculated threshold}
\nomenclature[P]{\(f_{\text{db}}, f_{\text{max}}\)}{Controller parameters (Hz)}
\nomenclature[P]{\(K_D\)}{Proportional coefficient for PD controller}
\nomenclature[P]{\(\Delta t\)}{Simulation time step (sec)}
\nomenclature[V]{\(\Delta f_{\text{nadir}}\)}{The maximum frequency deviation}
\nomenclature[V]{\(f_d^{\text{max}}\)}{The maximum rate of change or frequency during the contingency (Hz/sec)}
\nomenclature[P]{\(\eta_{\text{min}}\)}{ Portion of the fleet can be used for FFR}
\nomenclature[P]{\(P^{\text{nom}}\)}{Nominal power (MW)}
\nomenclature[P]{\(f_0\)}{Fundamental frequency for AGC signal (Hz)}
\nomenclature[V]{\(A\)}{Amplitude of AGC (MW)}
\nomenclature[V]{\(X^{0}\)}{Accepted requests when $c_0$ is tracked}
\nomenclature[V]{\(X^{h*}\)}{Accepted requests when $c_0+H_h$ is tracked}
\nomenclature[V]{\(X^{h}\)}{Accepted requests when $H_h$ is tracked }
\nomenclature[V]{\(q^{\text{dis}}\)}{Number of accepted discharging requests}
\nomenclature[V]{\(n_{u}\)}{Average number of accepted requests}
\nomenclature[P]{\(F\)}{Safety factor}
\nomenclature[V]{\(\rho\)}{Probability of violating safety factor}
\nomenclature[V]{\(n_{\text{ess}}\)}{Total number of batteries}
\nomenclature[V]{\(n_{\text{tcl}}\)}{Total number of TCLs}


\section{Introduction}\label{sec:intro}

The penetration of renewable energy sources (RES) such as wind and solar is increasing rapidly as a part of the global effort to reduce greenhouse gas emissions. However, increased RES use leads to more variability in electricity generation due to RES' intrinsic uncertainty. In addition, replacing conventional synchronous generators with inverter-based renewable generation reduces the inertia, i.e., the power system's ability to oppose changes in frequency~\cite{ercot2018}. 
A significant deviation from nominal frequency can lead to an outage of generation units which subsequently causes further frequency deviation and, in severe cases, can result in a total system blackout. A recent example was the Texas blackout in 2021, which led to approximately 155 billion dollars in loss \cite{BUSBY2021erss}. Therefore corrective actions are crucial, to control the power system frequency. Frequency control mechanisms include primary, secondary, and tertiary frequency control.  Primary frequency control also called fast frequency response is largely automatic and instantaneous and occurs over the first few seconds following a grid disturbance event. The secondary frequency control brings the frequency back to the nominal value by adjusting the output of generating units within a few minutes after a frequency event. Tertiary frequency control restores the power reserve of the generators used for the secondary frequency control~\cite{ZHANG2020ER}. In this paper, the focus is on primary control since secondary and tertiary control do not significantly influence the transient frequency dynamics.

Distributed energy resources (DER) are widely considered an effective and scalable way to provide primary frequency control~\cite{guo2018cdc, poolla2019TPWRS}. DER coordination can be used to provide synthetic damping, which is defined as the percentage change in the total DER consumption in response to frequency change~\cite{ercot2018}, as well as inertia, improving the stability of the power system. Coordinating DERs to provide primary frequency control has been studied for many years~\cite{schweppe1980PAS,brokish2009,zhao2018,Mendieta2021}. One of the first works on frequency responsive loads was presented in 1980 called frequency adaptive power and energy re-scheduler (FAPER)~\cite{schweppe1980PAS}. In this method, the dynamic state of thermostatically controlled loads (TCLs) is used to prioritize TCLs for frequency response. For example, devices with high temperatures will be prioritized to be turned off during an under-frequency event. In addition, the bound on temperature are frequency-dependent meaning that for higher frequency deviations, more devices participate in the frequency response.  Probabilistic FAPER was introduced in~\cite{brokish2009} by injecting random delays in devices switching on/off, which helped address synchronization concerns with FAPER, i.e., avoided large groups of DERs attaining nearly the same temperature and, thus, responding nearly identically and causing large power swings.


\par
Different types of DERs can be coordinated for primary frequency control. TCLs (e.g., electrical water heaters (EWHs) and refrigerators) can be turned off for short periods without a considerable effect on the temperature, which provides some flexibility used for a  rapid change in load. Also, TCLs form a large portion of the power system load~\cite{zhao2018}, and therefore, coordinating them provides considerable capacity for the power system operator. Another significant advantage of using TCLs is that they are highly responsive, making them an appropriate option for primary frequency control where fast response is required. In~\cite{Clarke2020}, smart EWHs are used to compensate for the uncertainty in wind and solar energy, peak shifting, and frequency response. A dynamic model is presented in~\cite{Mendieta2021} for different types of TCLs that adapt and improve a direct load control (DLC) scheme for primary frequency regulation in hybrid isolated microgrids. 
Frequent on/off switching of TCLs increases wear-and-tear and should be minimized during TCL coordination schemes as discussed in~\cite{coffman2020ACC}.

 Different control architectures have been proposed for DER coordination, which differs in the level of communication requirements, quality of service (QoS), number of cycling, and level of grid awareness. In~\cite{Syed2018}, a novel method is proposed that relies on transient phase offset to achieve a fast response to primary frequency control. The method enables the power system operators to use resources closer to the frequency deviation source. Such a technique also ensures that the stability of the system is preserved. In~\cite{Weitenberg2019}, a fully decentralized leaky integral controller for frequency restoration is presented. The use of decentralized control led to the elimination of communication delays and failures (e.g., lost messages) associated with centralized control schemes. Instead,~\cite{Weitenberg2019} uses communication between local loads, which decreases the communication structure costs significantly. Other methodologies employ adaptive controllers that adapt to online measurements. One example of such techniques was introduced in~\cite{Jin2019}, where an adaptive control framework is built based on the time-space distribution characteristics of the frequency in the power system. Also, the frequency response control is transformed from decentralized feedback control to centralized feed-forward control. Moreover, adaptive controllers have been shown to reduce problems with actuation delay. Finally, hierarchical optimization-based DER coordination schemes were developed in~\cite{almassalkhi2020energies} with the advantage that AC network constraints can be considered. This is the so-called {\em grid-aware coordination}. 
 
 The authors have presented a fully decentralized proportional controller in~\cite{mavalizadeh2020SGC} to provide synthetic damping from a TCL fleet. In this manuscript, the results from~\cite{mavalizadeh2020SGC} are extended to arbitrary reference signals. While the proposed control scheme is tested on a timer-based prioritization scheme~\cite{almassalkhi2018chapter}, it can be applied to other {\em fitness-based} DER prioritization schemes (e.g.,~\cite{nandanoori2018CTA}) as well. Each DER measures frequency locally and based on the designed control law, a frequency-dependent threshold on the timer is calculated. Based on the calculated threshold, devices determine whether to participate in primary frequency control or not. To the best of the author's knowledge, this is the first work to provide an analytical estimate of the amount of synthetic damping that can be extracted from a DER fleet. The main contributions of this manuscript are listed below:
\begin{itemize}
    
    \item The decentralized energy packet interruption controller initially proposed in~\cite{mavalizadeh2020SGC} is now generalized to consider bi-directional DERs and rate-of-change of frequency (RoCoF) to improve the fleet's response during either under- or over-frequency events. 
    \item Using limited information available to the coordinator, the synthetic damping available from a fleet of packetized resources can be precisely estimated in real-time. 
    \item Spectral analysis of historical AGC data is used to develop and compute a probabilistic lower bound on the expected synthetic damping available from a fleet. This lower bound can be used to analyze the trade-off between a fleet's expected primary and secondary frequency control capabilities. 
    \item Simulation-based analysis is provided on practical considerations for packet-based DER coordination and synthetic damping, such as local measurement resolution and sensing/controller delays.
\end{itemize}

The remainder of the paper is organized as follows: a brief description of the dynamic grid model and packet-based coordination~is provided in section~\ref{sec:prelim}. The proposed control law is described in section~\ref{sec:controller}. In section~\ref{sec:tradeoff} the trade-off between FFR services and other ancillary services is characterized. Practical considerations are discussed in section~\ref{sec:practical}. Finally, section~\ref{sec:conclusion} provides the conclusion.

\section{ Preliminaries}\label{sec:prelim}
In this section, the network model is presented and the concepts of packet coordination and packet interruption are provided.

\subsection{Network Model}
     Let $G=(\mathcal{V},\mathcal{E}$) be a graph representing the topology of a transmission network, where $\mathcal{V}:=\{1,...,N_n \}$, is the set of $N_n$ nodes and $\mathcal{E} \subseteq \mathcal{V}\times \mathcal{V}$ is the set of branches, such that if $i$ and $j$ are connected, then $(i,j) \in \mathcal{E}$. The frequency dynamics of the network are  governed by the swing equations~\cite{kundur}:
\begin{subequations} 
    \begin{small}
	   \begin{align}
         \Delta\dot{\theta}_j = & {\,} \Delta\omega_j,\\
        M_{j}\Delta\dot{\omega}_{j} = & {\,} \Delta P^{\text{G}}_{j}-\Delta P^{\text{L}}_{j}-\Delta P^{\text{DER}}_{j} -D_{j}	\Delta \omega_{j} + \sum_{i:(i,j) \in \mathcal{E}}^{N} \Delta P_{ij},  \label{eq:swingDyn2}
	   \end{align}
     \end{small}
\end{subequations}
where $\theta_j$ and $\omega_j$ are the voltage angle and angular velocity at bus $j$, respectively, and $\Delta P^{\text{G}}_{j}$, $\Delta P^{\text{L}}_{j}$, and $\Delta P^{\text{DER}}_{j}$ are deviations in the generation, uncontrollable load, and controlled DER from nominal at bus $j$, respectively. Inertia [sec] and damping [MW/Hz] are described by $M_j$ and $D_j$, respectively, while $P_{ij}$ denotes the power flow [MW] between areas $i$ and $j$, respectively. 
    The $j^{\text{th}}$ generator's turbine dynamics is modeled in~\eqref{eq:turbine},
	\begin{align}\label{eq:turbine}
	\tau_j \Delta \dot{P}^{\text{G}}_{j}= -\Delta P^{\text{G}}_{j}-\frac{\Delta \omega_{j}}{R_{j}}, 
	\end{align}
	where $R_{j}$ is the generator's governor droop coefficient [Hz/MW] and $\tau_j$ is the turbine time constant [sec]. Synchronous generator droop controllers usually have a deadband of 36 mHz~\cite{kirby2003ornl}. To model the DERs response, the basics of packet-based coordination are presented next.  
	
	\subsection{Packet-based DER coordination}\label{subsec:pem}
 	
     Packet-based DER coordination is enabled by a cyber-physical system that coordinates incoming and asynchronous discrete grid-access requests for energy from individual DERs~\cite{almassalkhi2018chapter, duffaut2020PES, brahma2022optimal, Almassalkhi2022spectrum}. The DER's asynchronous requests are central to packet-based coordination and are explained next.
    \begin{definition}{(Energy packet)}
    An energy packet is a fixed-duration and fixed-power epoch of energy consumed (or delivered) by a DER. 
    \end{definition}

 In packet-based schemes,  each DER requests an energy packet based on its need for energy (NFE). DERs considered in this manuscript represent residential water heaters, EV chargers and/or residential batteries, whose energy (e.g., temperature and SoC) dynamics are much slower than the frequency response. This means that during the primary frequency control period, the specific model of the DERs is not significant, as long as their power consumption is adjustable. For example, if an air conditioner (A/C) measures a room temperature in the summer above some desired set-point, then the room temperature is too high and the device's NFE increases. This leads to more frequent requests for energy packets to cool down the room. Similarly, if the A/C measures a low room temperature, the NFE decreases and the device will not request an energy packet often, if at all. The energy packet requests then arrive from devices asynchronously and each request is either accepted or rejected by the coordinator based on aggregate demand and a market or grid reference signal. When a request for an energy packet is accepted, an internal timer for the switched device is triggered and the DER charges or discharges until the timer's absolute value equals the packet length (or epoch length).  The local timer for DER $n$ is described as
\begin{align} \label{eq:timer}
	 t_{n}[k+1]= 
	        \left\{
        	 \begin{array}{ll} 
        	    t_{n}[k]+\Delta t,   & \text{if } C_n[k]= 1\\
        	    t_{n}[k]-\Delta t,   & \text{if } C_n[k]= -1\\
        	    0                    & \text{otherwise}  
        	\end{array}    
           \right.,
\end{align}
	where $\Delta t$ is the sampling time, the number of bins is $n_p :=\lfloor \delta/\Delta t_B \rfloor$ and $\Delta t_B$ is the timer bin width. The packet duration (epoch) is denoted $\delta$ and typically is set between 60 to 600 seconds. Without loss of generality, one can choose $\Delta t= \Delta t_B$. When the $n^{\text{th}}$ DER has its charging request accepted at time $k$, then $C_n[k+1]=C_n[k+2]=\dots C_n[k+n_p]=1$. On the other hand, if a discharging request packet is accepted at time $k$, $C_n[k+1]=C_n[k+2]=\dots C_n[k+n_p]=-1$. Otherwise, $C_n[k+1]=0$. In fact, $C_n$ is 1 when the device is charging, -1 when the device discharges, and 0 when the device is OFF. 
 
 Even though the coordinator does not have access to the individual DERs' internal timers, it knows how many requests were accepted at each time step in addition to packet height $P^{\text{rate}}_n$ for each request, which permits the coordinator to construct an accurate estimate of the DER fleet's aggregate timer. In general, the coordinator needs to consider four different timers using~\eqref{eq:timer}; $i)$ Charge-only timer which includes devices that only can be charged, such as TCLs and ACs; $ii)$ Discharge-only timer which includes devices that can only be discharged, such as solar panels and stand-alone gen-sets; $iii)$ Charging bi-directional devices, such as ESS; $iv)$ discharging bi-directional devices.  In this paper, the focus is on DER fleets of TCLs and ESSs, therefore, its three corresponding timers are $(i,iii,iv)$. Since all accepted DERs start their packet at the first bin, the linear timer dynamics is defined by
	 \begin{align}  \label{eq:histogram}
    	\notag x^{\text{ch}}_{
    	\text{tcl}}[k+1] &= Mx^{\text{ch}}_{\text{tcl}}[k]+Bq^{\text{ch}}_{\text{tcl}}[k],\\
    	 x^{\text{ch}}_{\text{ess}}[k+1] &=Mx^{\text{ch}}_{\text{ess}}[k]+Bq^{\text{ch}}_{\text{ess}}[k],\\
     	\notag x^{\text{dis}}_{\text{ess}}[k+1] &= Mx^{\text{dis}}_{\text{ess}}[k]+Bq^{\text{dis}}_{\text{ess}}[k],
        \end{align}
        where $x^{\text{ch}}_{\text{tcl}},x^{\text{ch}}_{\text{ess}}, x_{\text{ess}}^{\text{dis}} \in \mathbb{R}^{n_p}$ are binned distributions of power for charge-only TCLs, bidirectional charging ESS, and bidirectional discharging ESS, respectively, while  $q^{\text{ch}}_{\text{tcl}}[k], q^{\text{ch}}_{\text{ess}}[k], q^{\text{dis}}_{\text{ess}}[k] \in \mathbb{R}$ are the total power of accepted charging TCL, charging ESS, and discharging ESS requests during time step $k$, respectively. That is, $q^{\text{ch}}_{\text{tcl}}[k]=\sum_{n\in I[k]} P^{\text{rate}}_n$, where $I[k]$ is the set of DERs with accepted requests at time $k$. The timer dynamics are defined by $M \in \mathbb{R}^{n_p \times n_p}$, which is a lower triangular matrix with 1's on the lower off-diagonal and zero elsewhere, while $B := [1, 0,\hdots, 0]^\top \in \mathbb{R}^{n_p}$. Thus, when a request is accepted, the accepted DER enters the first bin, and at each time step it propagates through the timer. The number of devices completing their packets at time-step $k+1$ is equal to the number of devices in the last bin of the timer distribution. During a frequency event, the distribution can be considered constant because the timer states evolve slower than the grid frequency. That is, if a frequency event occurs at $k$, states  $x^{\text{ch}}_{\text{tcl}}[k],x^{\text{ch}}_{\text{ess}}[k], x_{\text{ess}}^{\text{dis}}[k]$ can be assumed constant since frequency response is a fast event (i.e., $<10$s).
    
	It is clear that the timer states are a function of past coordinator packet acceptance rates. 
    During the frequency event, packets actively participate in the frequency response based on their internal timer states. Therefore, the concept of packet participation is presented next by extending the packet interruption defined in~\cite{mavalizadeh2020SGC}. 
	\begin{definition}{(Packet participation)}\label{def:participation}
	 The packet's participation is the forced change of DER $n$'s local state $C_n[k]$ before the end of its epoch length (i.e., $t_n < \delta$) due to a frequency deviation event. 
    \end{definition}	
    
   To explain the role of packet participation, first, consider the power draw for a general DER $n$. Its power consumption at time step $k$ is $P_{n}[k] \in [\underline{P}_n,\overline{P}_n]$, where $\underline{P}_n=0$ for TCLs and $= -P_n^\text{cap}<0$ for (discharging) ESS and $\overline{P}_n=P_n^{\text{cap}}>0$ when ON (TCL) or charging (ESS). $P_n^{\text{cap}}$ is the power rating of DER $n$. DER $n$ is then participating in FFR, if, for example, during an under-frequency event, DER $n$ changes its consumption from $\overline{P}_n$ to $\underline{P}_n$. It is important to note that the formulation generalizes to the case when the DER reduces its power to a value larger than $\underline{P}_n$, as long as the coordinator is aware of the DER's available power change. If there is no ESS with $\underline{P}_n < 0$, packet participation refers to packet interruption as defined in~\cite{mavalizadeh2020SGC} for TCLs only. 
  The coordinator continuously monitors three distinct timers in real time. These timers encompass the binned power values of charging TCLs, charging bi-directional DERs, and discharging DERs, as illustrated in~\eqref{eq:histogram}. By adding the values across the timer bins, the total power within each timer is obtained. To determine the total consumption of the entire fleet at time $k$, the total discharging power is subtracted from the total charging power as shown below, 

   \begin{align}\label{eq:pder}
P^\text{DER}[k] := \mathbf{1}_{n_p}^\top (x^\text{ch}_\text{tcl}[k]+x^\text{ch}_\text{ess}[k]) + \mathbf{1}_{n_p}^\top x^\text{dis}_\text{ess}[k].
   \end{align}
        Based on measured grid frequency, the flexible (net) demand, $P^\text{DER}$, can then be actively modified via packet participation by selectively interrupting and/or ``toggling'' packets (e.g., charging at $P_n[k] = P_n^\text{cap} >0$ toggles to discharging at $P_n[k] = -P_n^\text{cap}<0$). The selection of which packets participate during any given frequency event will be based on a fully decentralized DER control law. 
	
	\begin{remark}
	Note that in some packet-based coordination schemes, devices can be interrupted before the completion of their packet to maintain quality of service, i.e. turned off in case of excessive temperature or turned ON in case of low temperature (also called opt-out) as discussed in~\cite{almassalkhi2018chapter}. In~\eqref{eq:histogram} the number of opt-outs is assumed negligible, which is reasonable when the DER fleet operates near nominal demand. The coordinator can further ensure that this assumption is valid by 
 constraining operations to only track power reference signals close to its fleet's nominal power. That is, tracking a signal with a relatively large amplitude can  cause the devices to deviate from the set point, which in turn leads to more opt-outs and interruptions.
	\end{remark}
 
	The proposed decentralized frequency control scheme is provided in the next section, where the DERs leverage information about their local timer state to participate in primary frequency control.

\section{Proposed decentralized control law}\label{sec:controller}
This section presents a fully decentralized control law, which prioritizes DER participation based on a local timer and dynamic state. The designed controller creates additional damping from DER fleet which is added to the conventional system damping, i.e., $D_j$ in equation~\eqref{eq:swingDyn2}. 

\subsection{Local DER control law}\label{subsec:controllaw}
 The overall layout of the proposed controller is shown in Fig.~\ref{fig:flowchart}. 
\begin{figure}
    \centering
    \includegraphics[width=1\columnwidth]{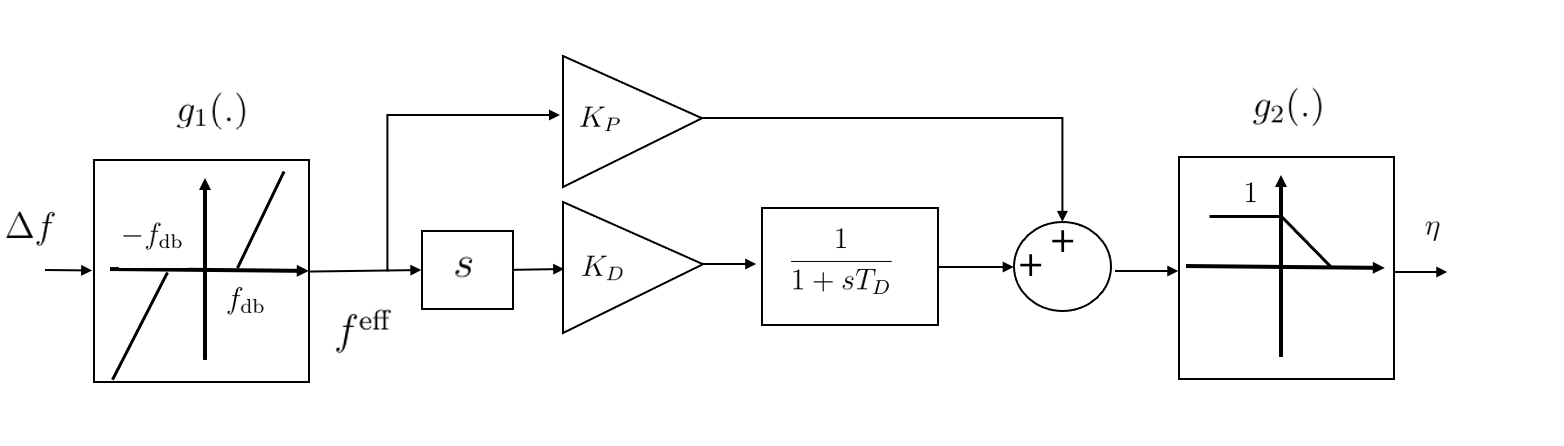}
    \caption{Block diagram of proposed derivative-proportional control law. $g_1(.)$ and $g_2(.)$ are defined in~\eqref{eq:g12}}.
    \label{fig:flowchart}
\end{figure}
A deadband with a size of $ f_{\text{db}}$ is defined such that if $|\Delta f[k]| <  f_{\text{db}}$ DER does not respond to a frequency deviation. $f_{\text{db}}$ is assumed 36 mHz to match a typical synchronous generator droop controller's deadband. A fleet's maximum participation is achieved at $f_{\text{max}}$ when all of the active devices have participated in primary frequency control. 
$\eta_{\text{min}}$ determines the portion of the timer that does not participate in FFR. In other words, devices will be locked for $\eta_{\text{min}}\delta$ seconds after their request is accepted.
The proposed local control law is:
{\small 
\begin{align}\label{eq:eta}
\begin{split}
& \eta_n[k]= \begin{cases} 
    	1, & |\Delta f[k]|<f_{\text{db}}\\ 
    	g_2\left(K_P f^{\text{eff}}[k]+K_D \mathcal{D} (f^{\text{eff}}[k])\right) , & f_{\text{db}}\le |\Delta f[k]| \le f_{\text{max}},\\ 
    	\eta_{\text{min}} & |\Delta f[k]| > f_{\text{max}}
    	\end{cases}
    	\end{split}
\end{align}
} where $K_P=1/(f_{\text{max}}-f_{\text{db}})$ and $K_D$ are the design parameters representing proportional and derivative gains, respectively, and $T_D$ is the derivative time constant. The derivative in the Laplace domain is denoted by $s$. $f^{\text{eff}}[k]:=g_1(\Delta f[k])$ with $\Delta f[k]$ the deviation from nominal frequency (e.g., 60 Hz) and functions
\begin{align}\label{eq:g12}
\begin{split}
  &  g_1(x):=\begin{cases}
    \mbox{0,} & \mbox{$|x|\le f_{\text{db}}$}\\
    \mbox{$|x|-f_{\text{db}},$} & \mbox{$ f_{\text{db}}\le|x|\le f_{\text{max}}$}
    \end{cases},
    \\
    & g_2(x):=\min\{\max\{-|x|+1, \,\eta_{
    \text{min}}\}, \,1\}.
    \end{split}
\end{align}
Note that $\mathcal{D}(f)$ is the backward discrete-time difference operator as expressed in~\eqref{eq:discderivative}. It estimates the rate of change of frequency (RoCoF) over the standard 500~ms window~\cite{nerc2015}.
\begin{align}\label{eq:discderivative}
    \mathcal{D}(f[k])=\frac{f[k]-f[k-\alpha_w/\Delta t_s]}{\alpha_w}
\end{align}
where $\alpha_w$ is the window size in [sec]. 
Thus, each device calculates its $\eta_n$ based on the locally measured frequency and participates in the frequency response, if $t_n[k]/\delta\ge \eta_n[k]$ for charging packets and $-t_n[k]/\delta\ge \eta_n[k]$ for discharging packets.
Immediately after a typical frequency event, the frequency deviation is zero while the magnitude of RoCoF is largest (i.e., $\mathcal{D}(f^{\text{eff}})$ is a monotonic function and $\mathcal{D}(f^{\text{eff}})$ approaches 0 exponentially). 
Therefore, the aggregate power response is initially due to the differential term $K_d \mathcal{D}(f^{\text{eff}})$ in~\eqref{eq:eta}. However, since $\mathcal{D}(f^{\text{eff}})\rightarrow 0$ exponentially (top plot in Fig.~\ref{fig:different_pd}), the proportional term $K_p f^{\text{eff}}$ becomes dominant, resulting in a linear decrease in aggregate power with respect to the frequency deviation. Finally, after the frequency reaches its nadir point and starts to recover, no more DERs participate, as shown in the bottom plot of Fig.~\ref{fig:different_pd}. As illustrated in Fig.~\ref{fig:different_pd}, the local control law can effectively coordinate packet participation at scale to improve the frequency response with higher $K_D$ values leading to more responsive (and aggressive) DER participation. 
Some remarks on controller tuning are presented in Subsection~\ref{subsec:tuning}.  
The next subsection makes use of the timer definition and the proposed control law to determine the available synthetic damping in real-time. 

\subsection{Real-time estimation of damping}

Since the coordinator determines how many devices are accepted during each time step and the packet height, i.e., $P^{\text{rate}}_n$ is known for any packet request, $x^{\text{ch}}_{
    	\text{tcl}}[k], x^{\text{ch}}_{
    	\text{ess}}[k]$, and  $x_{\text{ess}}^{\text{dis}}[k]$ can be accurately estimated by the coordinator in (effectively) real-time. Furthermore, to overcome any inaccuracies associated with the communication or actuation delays, 
     the coordinator can use feedback in the form of a simple acknowledgment sent (asynchronously) from each device when its operating state transitions. In addition, the frequency of the system can be measured by the coordinator, and from~\eqref{eq:eta}, a single $\eta[k]$ can be calculated for the entire fleet in real time. Moreover, given that the number of ON devices in each timer bin is known, the coordinator can then determine the available load reduction in response to frequency deviation without the need for additional communication with devices. For example,~\eqref{eq:rt} determines the amount of available power for an under-frequency event. 

        	\begin{align}\label{eq:rt}
    	   \nonumber 
            \Delta P^{\text{DER}}[k]=  &\sum_{i=1}^{\Tilde{K}}\left(x^{\text{ch}}_{\text{tcl}}[k-\delta/\Delta t +i-1]\right. \\ 
            & \left.+ 2x^{\text{ch}}_{\text{ess}}[k-\delta/\Delta t +i-1]\right),
    	\end{align}
    	where $\Tilde{K}:=\lfloor{\eta[k]\delta/\Delta t}\rfloor$. Thus, from the aggregate fleet power and any potential system frequency event (i.e., a deviation with nadir $\Delta f_\text{nadir}$), the coordinator can simply and, in real-time, estimate the available synthetic damping from a DER fleet as 
    	\begin{align}\label{eq:Dsyn}
    	    D^{\text{syn}}[k] = \frac{\Delta P^{\text{DER}}[k]}{\Delta f_{\text{nadir}}-f_{\text{db}}}. 
    	\end{align}
    
As seen in Fig.~\ref{fig:different_pd}, after reaching the nadir frequency, the frequency begins to recover, which results in an increase in $\eta_n[k]$ according to~\eqref{eq:eta}. However, it is important to highlight that even as $\eta_n[k]$ increases, DERs that have already been interrupted will not switch back on again. Therefore, the amount of damping is determined by the nadir frequency, as described in~\eqref{eq:Dsyn}. Fig.~\ref{fig:real_time} illustrates the accuracy of the synthetic damping estimate compared with actual damping provided by the DER fleet for 10 different frequency events. The top figure shows the change in power versus the change in frequency for one of these realizations. 
    	\begin{figure}
    \centering
    \includegraphics[width=1\columnwidth]{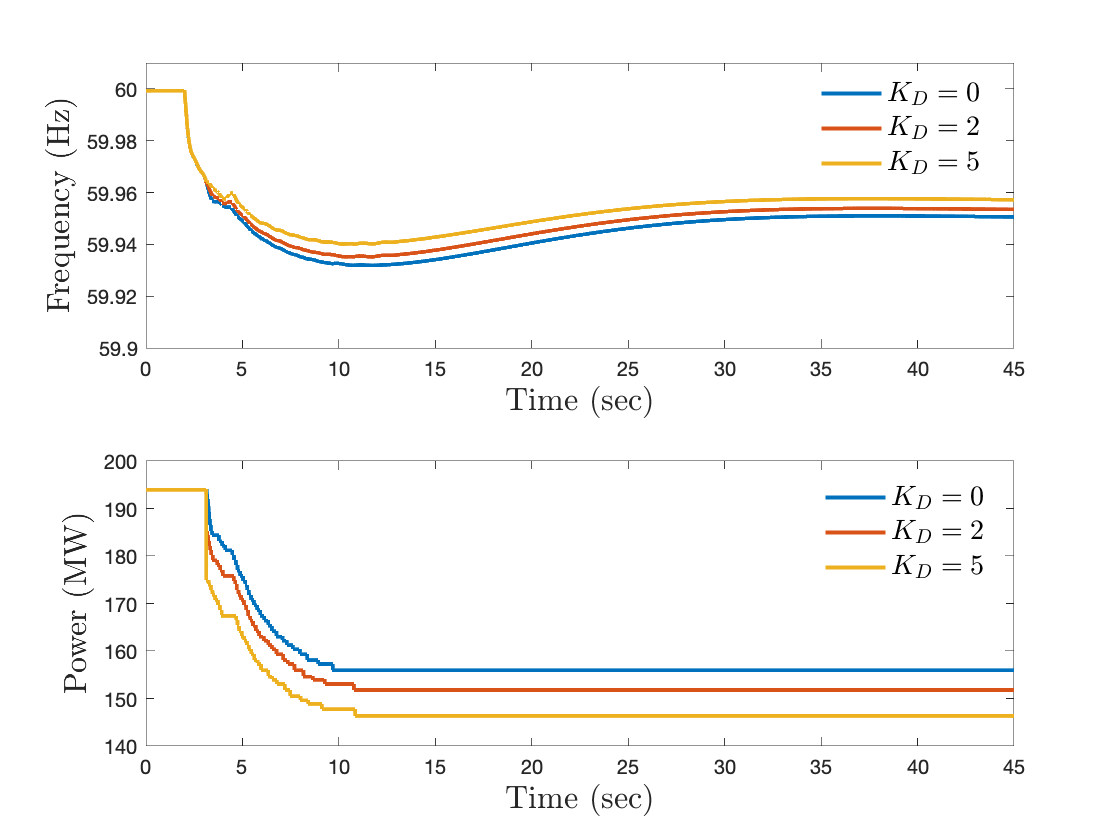}
    \caption{Frequency response of 200,000 TCLs for different values of $K_D$.}
    \label{fig:different_pd}
\end{figure}
    	\begin{figure}
    \centering
    \includegraphics[width=\columnwidth]{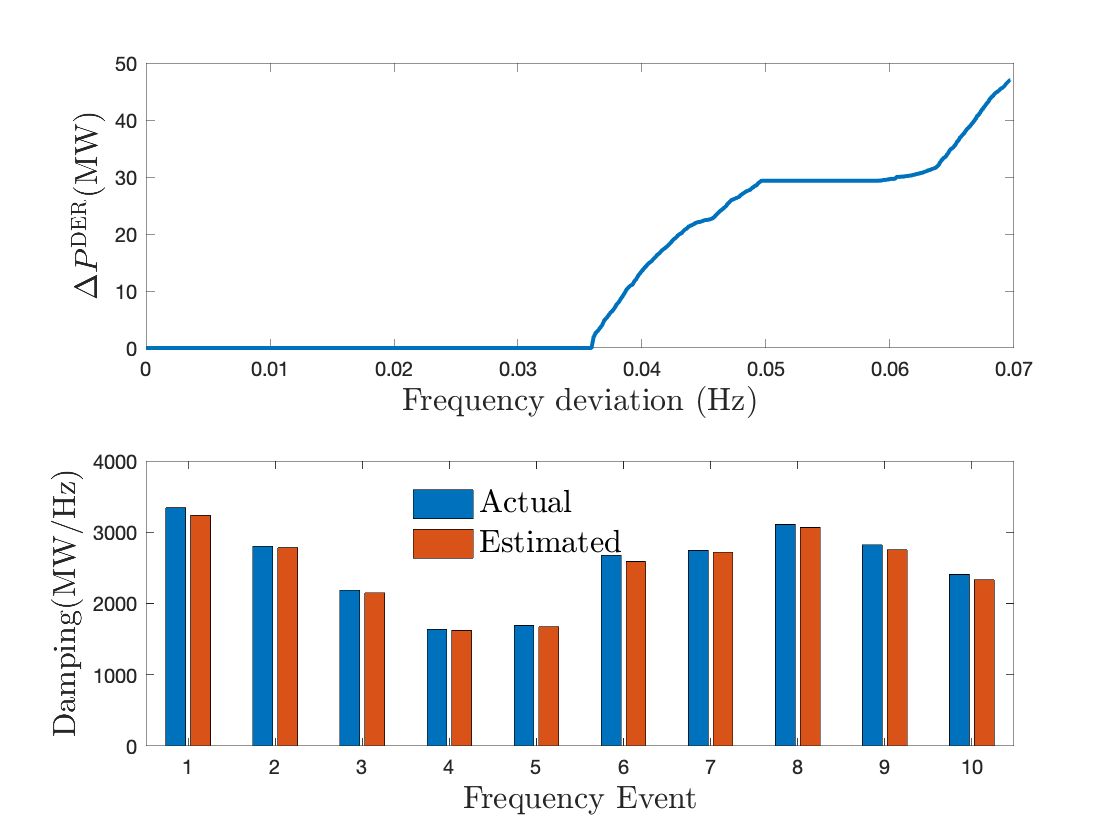}
    \caption{(\textit{Top}) A single frequency event with $\Delta f_\text{nadir} = 0.07$Hz yields a large load reduction in the aggregate DER fleet as a function of the frequency deviation, packet timer distribution, and designed control law. (\textit{Bottom}) Comparing the actual damping with the real-time estimate of synthetic damping for 10 different frequency events.}
    \label{fig:real_time}
\end{figure}

    Calculation time to find synthetic damping for a fleet of 200,000 DERs from Eq.~\eqref{eq:eta}, \eqref{eq:rt} and \eqref{eq:Dsyn} takes less than 500 $\mu s$ which is acceptable for real-time applications. This allows the coordinator to accurately and continuously estimate the available synthetic damping across a number of what-if scenarios (e.g., different frequency nadir and ROCOF pairs). Clearly, 100s of these calculations could be executed every 60 seconds to provide the coordinator/ISO with an accurate and real-time estimate of available synthetic damping capability from a fleet of DERs. The ISOs could then use these data/capabilities to evaluate stability margins/contingencies/ancillary services needs/etc.   	

In this section, a method is provided for the coordinator to accurately estimate the synthetic damping available from a fleet of DERs in any real-time operating condition (i.e., with an arbitrary, but known timer distribution). However, the coordinator may want to know a day- or hour ahead how much synthetic damping will be available from the fleet, in which case the timer distribution is unknown. Thus, in the next section, a probabilistic prediction of the available synthetic damping from a fleet of DERs participating in frequency regulation (e.g., PJM's Reg-D) is developed that captures a range of operating conditions via the amplitude of the Reg-D regulation signal. The method in~\cite{mavalizadeh2020SGC} is generalized to a bidirectional fleet by leveraging a specific packet-acceptance policy from~\cite{duffaut2020TCST} to ensure a unique mapping between the reference signal and the power-timer distribution. This enables an analytical characterization of synthetic damping statistics. It also can toggle the load between charging and discharging modes which double the synthetic damping available. Consequently, this permits us to analyze the trade-off between expected primary (damping) and secondary frequency control capabilities and (statistically) guarantee a DER fleet's ability to deliver synthetic damping.

\section{Characterizing the available synthetic damping}\label{sec:tradeoff}

 Here, a probabilistic framework incorporating historical AGC data is used to characterize the available synthetic damping that can be guaranteed (i.e., a lower bound) from a fleet of DERs that are coordinated via ON/Charge/Discharge packets and providing a certain MW-level of frequency regulation (AGC) services. The process of mapping a representative AGC signal from the fleet's timer distribution to the change in the fleet's aggregate power due to a frequency event is outlined next. It is based on a spectral decomposition of historical AGC (PJM Reg-D) data and was inspired by~\cite{brahma2022IREP, PNNL2011}.
 

To map timer states to changes in power, $\Delta P^\text{DER}$, for a given frequency event, the (conveniently) designed frequency-dependent timer threshold, $\eta$ is leveraged. 
\begin{remark}
The minimization policy, as presented in~\cite{duffaut2020TCST}, aims to track the reference signal with the fewest number of DERs. This policy guarantees that the number of DERs in the timer is always lower compared to other control policies. Consequently, it establishes a lower limit on the flexibility (i.e., synthetic damping) achievable for a given fleet. 
\end{remark}
However, the shape of the timer distribution is unknown in advance as it depends on the fleet's operating conditions (i.e., the reference signal and the number of available packet requests). Under the assumptions that (A1) a sufficient number of packet requests are available to the DER coordinator for effective aggregate power reference tracking (i.e., negligible tracking error); (A2) a fixed packet-request-acceptance policy (e.g, minimize the number of accepted packets) is adopted~\cite{duffaut2020TCST}; and (A3) the power reference signal is known ahead of time, then the exact timer distribution can be constructed over the duration of a packet epoch and the available synthetic damping can be estimated. However, if the coordinator wants to \textit{predict} the available synthetic damping ahead of time (e.g., for possible FFR markets or planning studies), the exact AGC power reference will be unknown (i.e., assumption A3 will not hold), which implies that the timer distribution will be unknown. To overcome this challenge, historical AGC data is used to characterize the statistics of the timer distribution and thus to provide a probabilistic lower bound on the synthetic damping availability. Hence, the methodology for characterizing the synthetic damping availability from a fleet of DERs consists of the following steps:
\begin{enumerate}[label=(\Alph*)]
    \item Decompose historical AGC signals into its $N$ most salient harmonics.
    \item From the spectral decomposition and under assumptions A1 and A2, determine the statistics of the corresponding timer distribution.
    \item Using the timer distribution statistics, determine a probabilistic lower bound on the number of packets in each timer. 
    \item Compute the probabilistic lower bound on the available synthetic damping from the fleet.
\end{enumerate}

\subsection{Spectral decomposition of AGC data}
Following~\cite{PNNL2011}, where a clustering technique is able to categorize similar 2-hour samples of historical AGC data based on spectral analysis, a 2-hour AGC sample is considered. That is, the \textit{methodology} is representative of a historical 2-hour AGC signal. However, the methodology can readily be applied to larger AGC data sets via methods presented in~\cite{PNNL2011}.

Consider the spectral (Fourier-based) decomposition of a 2-hour-long historical AGC sample signal of 2-second resolution into its $N$ most salient harmonics as
\begin{align}\label{eq:agc}
    AGC[k] \approx A \sum_{h=1}^{N} H_h[k],
\end{align}
where $H_h[k]:=c_h\cos(2\pi h f_0k-\phi_h)$. The coefficient $c_h$ for the $h$-th harmonic is scaled between 1 and -1, while phase shift is denoted by $\phi_h$, and $f_0$ is the signal's fundamental frequency. $A$ is the amplitude of the reference signal in MW. The reference signal is then defined as $P^{\text{ref}}[k]=P^{\text{nom}}+AGC[k]$ where $P^{\text{nom}}$ is the power set-point that maintains the average SoC of the fleet stationary and $AGC[k]$ is obtained from~\eqref{eq:agc}. The goal is to map $P^{\text{ref}}$ to the coordinator's timer distribution (under assumptions~A1 and A2). The procedure is detailed next for under-frequency events. The derivation for over-frequency events follows similarly. An expression for the total power of accepted requests at time $k$ for each harmonic $h$, $q^{+}_h[k]$, must be found first. $q^{+}_h[k]$ determines the power in the first bin of the timer at time $k$ for harmonic $h$. This is then used to find $q^{+}[k]$ which is the total number of accepted requests at time $k$.

It is convenient to decompose $H_h$ into two functions an increasing function ($Y_h[k]$) and a decreasing function ($Z_h[k]$). That is, $H_h[k]:= Y_h[k]-Z_h[k]$, where
		\begin{subequations}
		\begin{align}
		\begin{split}\label{eq:y1}
		Y_h[k]&=
		\begin{cases} 
		H_h[k], & \text{if}\;\; H_h[k]-H_h[k-1]>0\\ 
		0, & \text{otherwise}
		\end{cases}
		\end{split}
		\\
		\begin{split}\label{eq:y2}
		Z_h[k]&=
		\begin{cases} 
		-H_h[k], & \text{if}\;\; H_h[k]-H_h[k-1]<0\\ 
		0, & \text{otherwise}
		\end{cases}.
		\end{split}	
		\end{align}
				\end{subequations}
Under assumption A1, $Y_h[k]-Y_h[k-1]$ and $Z_h[k]-Z_h[k-1]$ define, respectively, the net increase and decrease in DER aggregate power reference signal at each time step $k$. Now, if enough packet requests are assumed available, the coordinator can accept enough of them to match this increase or decrease. Thus, at each $k$, each of the reference signal's harmonics can be matched with enough number of accepted packet requests entering the coordinator's timer. Note that when the fleet includes both charging and discharging requests, the mapping becomes non-unique (since charging and discharging packets can effectively {\em cancel} each other out). To ensure a unique mapping between (harmonic) reference signals and accepted requests, a minimizing packet accepting the policy at the coordinator is employed that essentially does not select both charging and discharging requests at the same time~\cite{duffaut2020TCST}.
 Thus, under assumptions A1 and A2, $q^+_h[k]$ is written as,
\begin{align}\label{eq:q+}
		&\nonumber q^{+}_{\text{h}}[k]=\\\nonumber &\left(f(Y_h[k]-Y_h[k-1])-f(Z_h[k-n_p]-Z_h[k-1-n_p])\right)\\
		&+\left(f(Y_h[k-n_p]-Y_h[k-1-n_p])\right),
		\end{align}
		where $f(x)=x$ for $x\ge 0$ and $f(x)=0$, otherwise. In~\eqref{eq:q+}, $f(Y_h[k]-Y_h[k-1])$ determines the increase in the reference signal while the second and third terms determine the number of expired discharging and charging packets, respectively. 
		
		The next theorem characterizes the statistics of the charging timer for the minimization policy.
	\begin{theorem}\label{thm:1}
	Let $\delta$, $\Delta f$, $P^{\text{cap}}$, and $n_{\text{p}}$, $\eta_{\text{min}}$, $f_{\text{max}}$, $f_{\text{db}}$ and $K_D$ be fixed for a given fleet under decentralized control policy~\eqref{eq:eta}. The mean and standard deviation of $q^{+}$ are given by: 
\begin{align}\label{eq:mean_q}
 	 &\mathbb{E}(q^{+})=n_u+\frac{N_dA\sum_{h=1}^{N}hc_h\Delta t}{T},\\ 
 	&\sigma^2(q^{+})=\sum_{h=1}^{N}\left(\frac{1-e^{\frac{-T^2}{6h^2}}}{2}(2\pi \Delta t f_hAc_h)^2 +\frac{2N_dA^2hc_h^2\Delta t}{T}\right)\label{eq:std_q},
    \end{align}
    where $n_u=P^{\text{nom}}/(n_p)$, $N_d=2$ for ESS fleet and $N_d=0$ for TCL fleet.
    \end{theorem} 
\begin{proof}
The proof is done by construction. The properties of $Y_h$ and $Z_h$ are utilized to simplify~\eqref{eq:q+} and find the mean and standard deviation. From~\eqref{eq:y1} and~\eqref{eq:y2} and the definition of $f$, it can be seen that
	\begin{align}\label{eq:fy}
	\begin{split}
\lefteqn{\hspace*{-0.6in}f(Y_h[k]-Y_h[k-1])=}\\
	&\begin{cases} 
	0, & k \in \frac{iT}{2}\;\; i=\pm 1, \pm 2, \dots \\ 
	Y_h[k]-Y_h[k-1], &\text{otherwise}
	\end{cases}
	\end{split}
	\end{align}
	In the case of the ESS fleet, the above equation implies that there are two discontinuities in $f(Y_h[k]-Y_h[k-1])$ during a period. That is $N_d=2$. In the case of a TCL fleet, since there are no discharging requests, $Z_h[k]=0$, and $H_h[k]=Y_h[k]$, which leads to $N_d=0$. Note also that $Y_h[k]-Y_h[k-1] = 0$ when $Y_h[k] =0$ and in any other case  
 \begin{align}
   Y_h[k]-Y_h[k-1] = & {\;} Ac_h\cos (2\pi f_h k\Delta t)-Ac_h\cos (2\pi f_h(k-1)\Delta t) \nonumber \\
   \approx & {\;} -2\pi f_h Ac_h\Delta t \sin (2\pi f_hk\Delta t), \label{eq:approx}
 \end{align}
 where $\omega_h = 2\pi f_h$.
	Using \eqref{eq:fy} and \eqref{eq:approx}, one can find the mean and variance of $q^+_h$ directly from the mean and variance of $\sin (2\pi f_hk\Delta t)$, where $k\sim \mathcal{U}(0, 1/f_h)$ represents the random time of failure. 
	The interest here is to find the mean and variance of $f(Y_h[k]-Y_h[k-1])$, which can be obtained using the corresponding mean and variance of $Y_h[k]-Y_h[k-1]$ given that these expressions only differ by $N_d$ over each period. Recalling that for a set $X_1, X_2,\dots, X_N$  of mutually independent normal random variables with corresponding means $\mu_1, \mu_2,\dots, \mu_N$ and variances $\sigma_1^2, \sigma_2^2,\dots, \sigma_N^2$ one has that
    \begin{align}\label{eq:Y}
		Y=\sum_{h=1}^{N}c_hX_h\sim \mathcal{N}\left(\sum_{h=1}^{N}c_h\mu_h,\sum_{h=1}^{N}c^2_h\sigma_h^2\right),
		\end{align}
  and assuming that the dependence between harmonics of the AGC decomposition is negligible, then the mean and variance of the timer variable $q^+$ can be obtained by adding the mean and variance of each harmonic. In addition, the expected value of the power of the accepted requests corresponding to tracking $P^{\text{nom}}$ is $n_u=P^{\text{nom}}/(n_p)$  \cite{mavalizadeh2020SGC}. Therefore, using~\eqref{eq:Y}, $q^{+}$ has a Gaussian distribution with mean and standard deviation given by~\eqref{eq:mean_q} and~\eqref{eq:std_q}.
\end{proof}

  \subsection{Finding probabilistic lower bound on available damping}\label{subsec:lowerbound}
    
Theorem 1 is now used to compute a lower bound on the total power in each bin of the timer, $P_{\text{min}}$, analytically. That is, $P_{\text{min}}$ is estimated to be
        \begin{align}\label{eq:nmin}
         & P_{\text{min}}=\mathbb{E}(q^{+})-F\sigma(q^{+}).
        \end{align}
 where $F$ is a {\em safety factor} determined by the operator. $P_{\text{min}}$ estimates, at each time step, the minimum power inside the timer. Higher $F$ leads to a more robust estimate but at the same time results in a more conservative estimate of the fleet's available damping. $F$ is usually defined by the information available on the underlying distribution of the uncertainty, e.g., distribution, statistics, etc. If more information is available about the distribution of uncertainty, less conservative estimation can be made. The probability of being within F standard deviations is defined as $\rho:=P(q^{+}[k]>\mathbb{E}(q^{+})-F\sigma(q^{+}))$. If no information about the distribution of $n_{\text{min}}$ is available (only mean and standard deviation are known), then using the Borel-Cantelli inequality~\cite{Bartolomeo2017}, $F$ is
     \begin{align}\label{eq:borel}
         F=\left(\frac{1-\rho}{\rho}\right)^\frac{1}{2}\cdot
     \end{align}
     The bounds obtained by Borel-Cantelli inequality are the worst-case scenarios and are unlikely to be encountered in practice. By assuming that there exists evidence that the distribution is unimodal, the Chebyshev generating function (CGF) can be used as shown below
     \begin{align}\label{eq:gcf}
         F=\left(\frac{1-\rho}{e \rho}\right)^\frac{1}{1.95}\cdot
     \end{align}
    \cite{stellato2014}. Furthermore, if the distribution is Gaussian, the safety factor is found as follows:
    \begin{align}\label{eq:gaussian}
        F=\sqrt{2}\;\;\text{erf}^{-1}(1-2\rho)\cdot
    \end{align}
 where $\rho$ is the probability of violation of the lower bound on $P_{\text{min}}$ calculated by~\eqref{eq:nmin}. By choosing the desired level for $\rho$ and based on the level of information available to the coordinator, $F$ is selected from the above equations.
 
In~\cite{mavalizadeh2020SGC} a method to calculate the synthetic damping for a uniform timer distribution was presented. By assuming that all of the timer bins are at $P_{\text{min}}$, one can use the method described in~\cite{mavalizadeh2020SGC}, to calculate a constant value for lower-bound damping. The next theorem is used to find the probabilistic lower bound on damping from $P_{\text{min}}$. 
    \begin{theorem}\label{thm:2}
    Given Theorem~\ref{thm:1}, for a given contingency with known $\Delta f_{\text{nadir}}\in(f_{\text{db}},f_{\text{max}})$ and RoCoF, the minimum damping for under-frequency events is,
    \begin{align}\label{eq:lowerbound}
	D^{min}=n_pP_{\text{min}}^{\text{eff}}.\left(K_P+\frac{K_{D} R^{\text{max}}}{\Delta f_{\text{nadir}}-f_{\text{db}}}\right).
\end{align}
\end{theorem}
where $P^{\text{eff}}_{\text{min}}=P_{\text{min}}$ for TCLs and $P^{\text{eff}}_{\text{min}}=2P_{\text{min}}$ for ESS.
    \begin{proof}
    The proof is by construction. From~\eqref{eq:eta}, the total change in the fleet's aggregate power for a uniform distribution is:
    \begin{align}\label{eq:deltap}
        \Delta P^{\text{DER}}=(1-\eta_{\text{nadir}})P_{\text{ON}}, 
    \end{align} 
    where $\eta_{\text{nadir}}$ is the calculated $\eta$ at nadir frequency. $P_{\text{ON}}$ is the total power of ON devices and is calculated as $n_p P_{\text{min}}$.
In under-frequency events, charging devices participate in frequency response, and discharging devices do not participate. From control law in~\eqref{eq:eta}, \eqref{eq:deltap}
     and the damping definition~\cite{ercot2018}, the estimated damping is $	D^{min}=\Delta P^{\text{DER}}/\left(\Delta  f_{\text{nadir}}-f_{\text{db}}\right)$, which leads to $n_pP_{\text{min}}\left(K_P+\frac{K_D R^{\text{max}}}{\Delta f_{\text{nadir}}-f_{\text{db}}}\right)$ for $ f_{\text{db}}\le |\Delta f_{\text{nadir}}|\le f_{\text{max}}$. As mentioned in~\ref{subsec:controllaw}, the magnitude of $\mathcal{D} (f^{\text{eff}}[k])$ is at its maximum at the beginning of the disturbance and it decreases exponentially. Therefore, it is possible to replace $\mathcal{D} (f^{\text{eff}}[k])$ with the known maximum RoCoF of the event, $R^{\text{max}}$. Since ESS charging devices can be toggled to discharging during the frequency event, they can provide twice their capacity. This is captured by $P^{\text{eff}}_{\text{min}}=P_{\text{min}}$ for TCLs and $P^{\text{eff}}_{\text{min}}=2P_{\text{min}}$ for ESS. If the frequency deviation is less than $f_{\text{db}}$, damping is zero as indicated by~\eqref{eq:eta}. Finally, if frequency deviation exceeds $f_{\text{max}}$, all of the available power is shed, resulting in $D^{\text{min}}=\left(P^{cap}n_pn_{\text{min}}\right)/\left(\Delta f_{\text{nadir}}-f_{\text{db}}\right)$.
\end{proof}
Observe that from Eqs.~\eqref{eq:borel},~\eqref{eq:gcf} or~\eqref{eq:gaussian} the probability of violation of $P_{\text{min}}$ ($\rho$) can be found based on the level of information available about the timer distribution. To relate $\rho$ to the probability of violation of $D^{\text{min}}$, the following remark is used.
\begin{remark}
It is straightforward to show that the probability of violation of the calculated minimum damping in Theorem~\ref{thm:2}, is always smaller or equal to the probability of violation of $P_{\text{min}}$ ($\rho$).
\end{remark}
Fig.~\ref{fig:rho} shows the estimated probability of violating the bounds for different information available versus the actual percentage of violations. 9 and 18 MW are chosen for $A$ which are equal to $5 \%$ and $10 \%$ of the nominal power. The green and blue curves show the percentage of violations obtained in simulations whereas the blue, red, and yellow dashed lines are obtained from Eqs.~\eqref{eq:borel},~\eqref{eq:gcf} and~\eqref{eq:gaussian}, respectively.
\begin{remark}
    It should be noted that the lower bound on damping is derived under the assumption of negligible tracking error.  Based on the previous work, an epoch length of 3~minutes or less satisfies this assumption~\cite{khurram2021powertech}. For higher amplitudes of AGC, the tracking error increases since the reference signal has higher fluctuations around the nominal power. Therefore, for higher AGC amplitudes, because of the higher tracking error, the probability of violation of the calculated bounds is higher. This can be seen in Fig.~\ref{fig:rho} by comparing the violation probabilities for A being 5\% and 10\% of the nominal power. 
\end{remark} 
Fig.~\ref{fig:lowerbound} shows the accuracy of the estimated damping versus the true damping for different amplitudes of AGC. A fleet of 200,000 EWHs  with 4.5 kW capacity each is used in a two-area power system. The blue curve shows the mean of the true damping for 100 different realizations, while the green and purple curves show the mean minus 1 and 2 standard deviations, respectively. The dashed line indicates the estimated lower bound on damping calculated by~\eqref{eq:lowerbound}. Similar results are given for an ESS fleet in Fig.~\ref{fig:lowerbound_ess} when $K_D=5$. It should be noted that for each fleet, the reference signal is scaled around the fleet's nominal power. Therefore, the reference signal used to generate figures~\ref{fig:lowerbound} and~\ref{fig:lowerbound_ess} is different which leads to a difference in the provided damping as seen in the figures.  
\begin{figure}
    \centering
    \includegraphics[width=1\columnwidth]{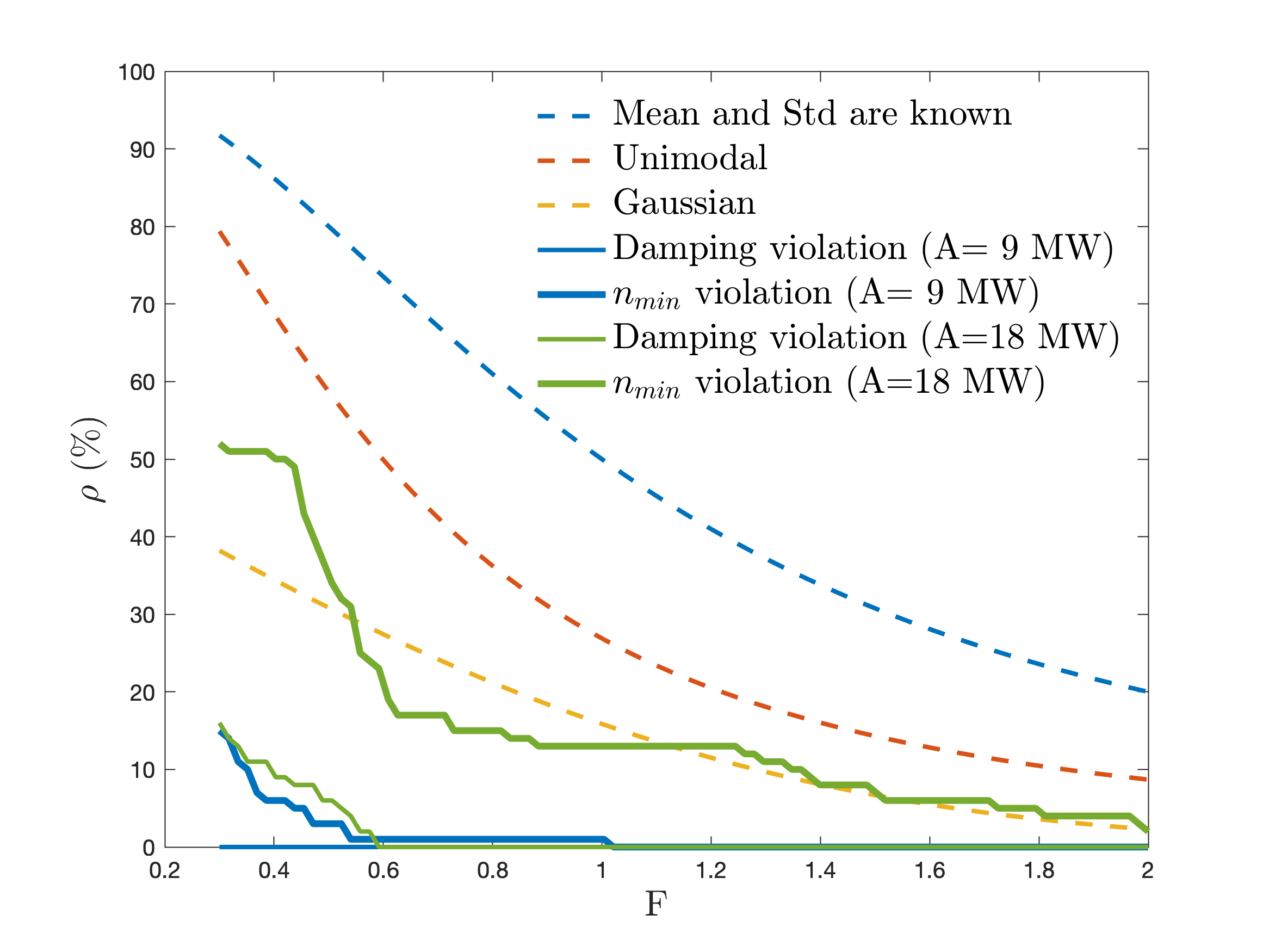}
    \caption{Probability of violating the bounds}
    \label{fig:rho}
\end{figure} 
\begin{figure}
    \centering
    \includegraphics[width=1\columnwidth]{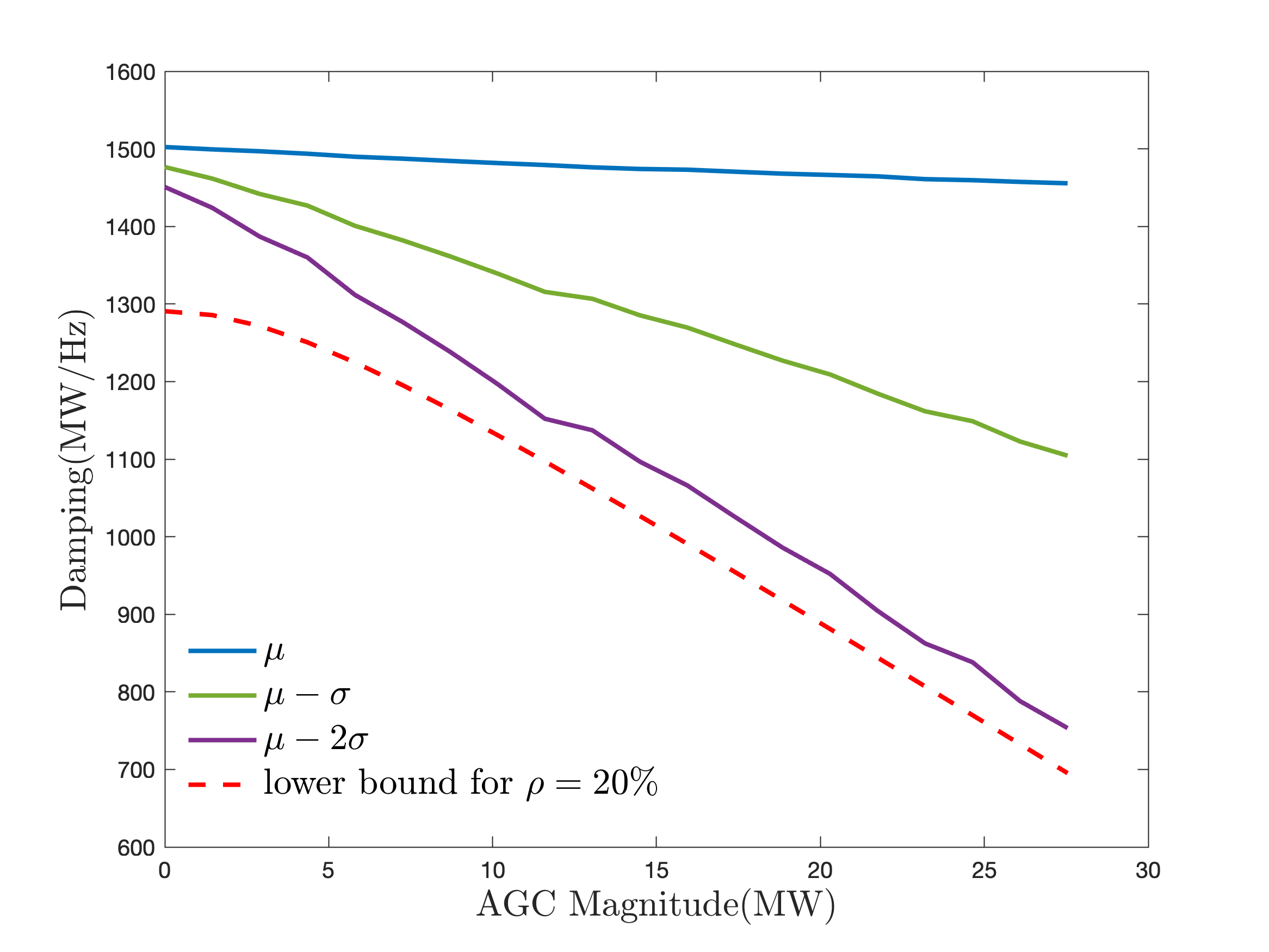}
    \caption{The actual vs estimate damping lower bound for 200,000 TCLs}
    \label{fig:lowerbound}
\end{figure}
\begin{figure}
    \centering
    \includegraphics[width=1\columnwidth]{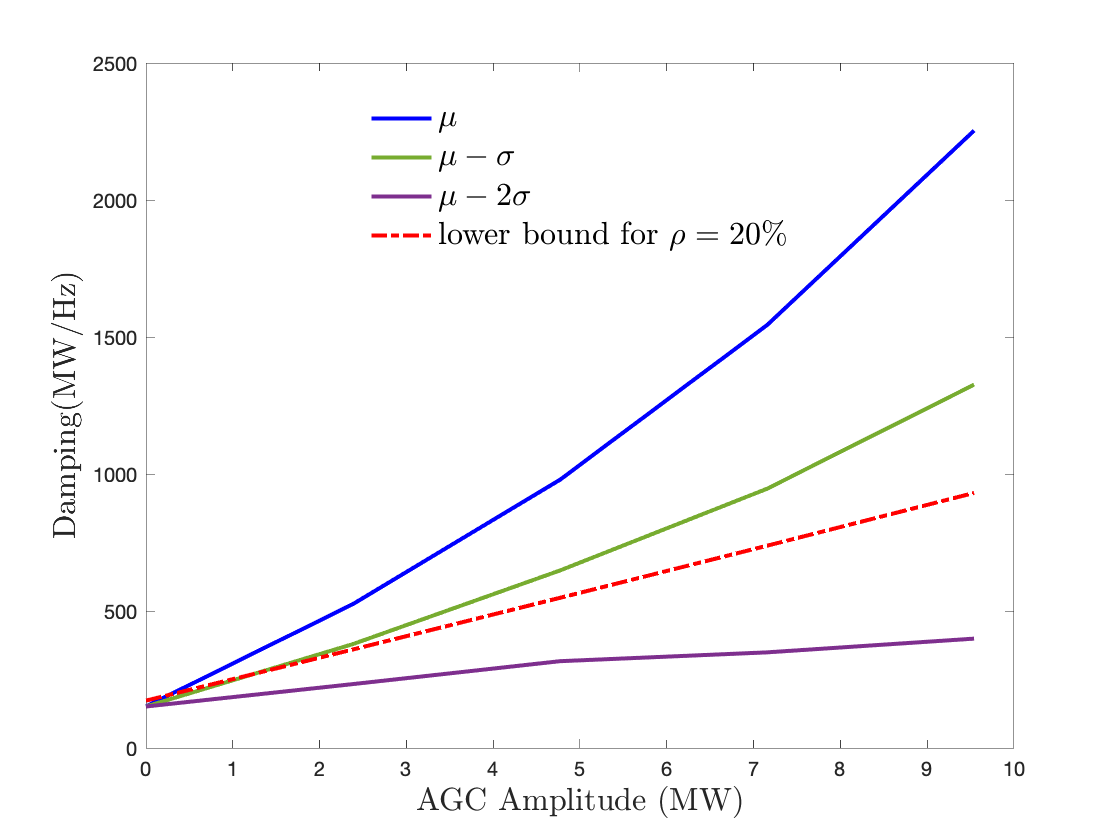}
    \caption{The actual vs estimate damping lower bound for 200,000 batteries}
    \label{fig:lowerbound_ess}
\end{figure}

To find the number of harmonics $(N)$ required to reconstruct the AGC signal a compromise between computation burden and accuracy must be taken into account. The goal is to find a $N$ that gives an acceptable reconstruction error for all of the 1-hour samples of the AGC data. To do so, the reconstruction error is calculated for all of the 1-hour samples in one year. The simulation results show that for $N=23, 35, 53$ the RMSE of reconstruction error is always lower than 30\%, 20\%, and 10\%, respectively. Therefore, by choosing $N=53$, it can be guaranteed that the reconstruction error for any day of the year is lower than 10\%. In this paper, $N=100$ is chosen which limits the construction error for any given day to 6\%.
\subsection{FFR versus frequency regulation trade-off}
From Theorem~\ref{thm:1},~\eqref{eq:nmin} and~\eqref{eq:lowerbound}, it can be seen that tracking a larger AGC signal (larger $A$) leads to higher variance. This always translates to lower $n_{\text{min}}$ and $D^{\text{min}}$ for TCL fleet. For the ESS fleet, since $\mathbb{E}(q^{+})$ is a function of $A$, as mentioned in Theorem~\ref{thm:1}, higher $A$ does not necessarily lead to lower synthetic damping. In this section, a procedure to determine the proper $A$, for a TCL fleet is presented to maximize the total profit. The same procedure can be applied to the ESS fleet, as well.\\
If the prices of frequency regulation is $\beta^{\text{Reg}}$ [\$/MW] and FFR damping is $\beta^{\text{FFR}}$ [\$/MW/Hz], then the total revenue can be written as $\beta^{\text{Reg}} A + \beta^{\text{FFR}}D^{\text{min}}$. If one defines $O:=A\beta + D^{\text{min}}$, where $\beta:=\beta^{\text{Reg}}/\beta^{\text{FFR}}$, then it is straightforward to show that maximizing $O$ maximizes the total revenue. By replacing~\eqref{eq:mean_q} and~\eqref{eq:std_q} in~\eqref{eq:nmin}, one gets \begin{align}\label{eq:nmin2}
P_{\text{min}}= n_u-2F\pi \Delta t f_0A\sqrt{\sum_{h=1}^{N}h^2c_h^2\left(\frac{1-e^{\frac{-T^2}{6h^2}}}{2} \right)}   
\end{align}
Now, substituting~\eqref{eq:nmin2} in~\eqref{eq:lowerbound} gives
\begin{align}\label{eq:O}
    \nonumber O&=A\beta+n_p\left(n_u- 2F\pi \Delta t f_0A\sqrt{\sum_{h=1}^{N}h^2c_h^2\frac{1-e^{\frac{-T^2}{6h^2}}}{2}}\right)\\
    &\cdot \left(K_P+\frac{K_D R^{\text{max}}}{\Delta f_{nadir}-f_{db}}\right).
\end{align}
To maximize $O$ in~\eqref{eq:O}, the derivative with respect to $A$ is calculated as $\frac{\partial \text{O}}{\partial A} = \beta - \beta^{\text{thr}} $, where
\begin{align}\label{eq:betamax}
    \beta^{\text{thr}} =n_p 2F\pi \Delta t f_0\sqrt{\sum_{h=1}^{N}h^2c_h^2\frac{1-e^{\frac{-T^2}{6h^2}}}{2}} \left(K_P+\frac{K_D R^{\text{max}}}{\Delta f_{nadir}-f_{db}}\right)  \cdot
    \end{align}
    From~\eqref{eq:betamax}, it can be seen that $\beta^{\text{thr}}$ is a function of frequency events. Then, for a set of credible system contingencies $\{\mathcal{C}_1,\dots,\mathcal{C}_{N_c}\} \in \mathcal{C} $, with known $\Delta f_{nadir,c}$, $R^{\text{max}}_c$ and probability $w_c$, $\beta^{\text{thr}}_1,\dots,\beta^{\text{thr}}_{N_c}$ can be calculated from~\eqref{eq:betamax}. Finally, a weighted average for $\beta^{\text{thr}}$ is calculated as follows:
    \begin{align}\label{eq:bmaxmean}
    \overline{\beta^{\text{thr}}}=\sum_{c=1}^{N_c} w_c\beta^{\text{thr}}_c.
    \end{align}
The normalized revenue, $O [MW/Hz]$, is presented in Figure~\ref{fig:economic} for different $\beta$ and $A$.  Total revenue in dollars [\$] is obtained by $\beta^{\text{FFR}}O$. 
\begin{figure}[t]
\centering
\includegraphics[width=1\linewidth,trim={0 0 0 0.5cm},clip]{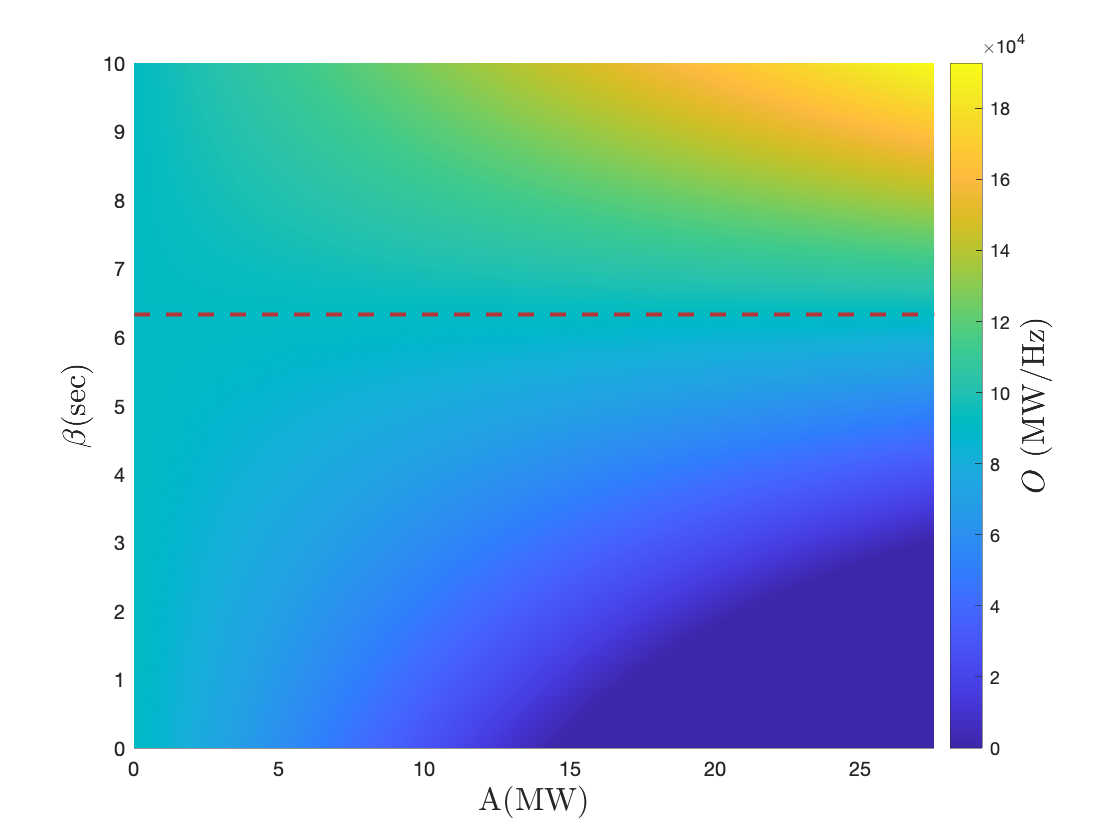}.
  \caption{The impact of ancillary service prices on the total normalized revenue, $O [MW/Hz]$. The red line indicates $\beta^{\text{thr}}$ above which larger AGC magnitudes, increase revenue.}
\label{fig:economic}
\end{figure}
The following conclusions can be derived from~\eqref{eq:O},~\eqref{eq:betamax}, and~\eqref{eq:bmaxmean}:
\begin{enumerate}
    \item  If $\beta \ge \overline{\beta^{\text{thr}}}$, $\frac{\partial \text{O}}{\partial A} > 0$, which means that to maximize the profit, the largest possible value for $A$ must be chosen,  
    subjects to the fleet being able to track without a considerable increase in the tracking error. 
    \item If $\beta < \overline{\beta^{\text{thr}}}$, $A^{\text{opt}}=0$. 
    \end{enumerate}

In the proposed method, DERs are used in two modes. When the frequency is close to nominal (i.e., $|\Delta f[k]|<f_{\text{db}}$) the fleet tracks the AGC signal, and when frequency deviation exceeds $f_{\text{db}}$ fleet goes to primary frequency control mode. Larger AGC amplitudes mean that more resources are used for AGC tracking. It can be seen from Fig.~\ref{fig:lowerbound} that for a TCL fleet, this leads to higher variation in damping which limits the ability to guarantee a minimum value for damping. Therefore, a fundamental trade-off between primary and secondary frequency control capability exists as seen in~\eqref{eq:nmin2}.

Another observation from~\eqref{eq:betamax} is that by using derivative control (i.e., increasing $K_D$), $\overline{\beta^{\text{thr}}}$ increases. This means that providing primary frequency control will be profitable for lower FFR prices ($\beta^{\text{FFR}}$).

\subsection{Tuning of the controller parameters}\label{subsec:tuning}
The coordinator is assumed to only have access to its own DER information and not that of the grid operators or other coordinators. As such, the tuning of $K_\text{P}, K_\text{D}$ is based on the coordinator's fleet information and published system-wide reliability metrics, such as frequency nadir and RoCoF. For credible contingencies in the system, the initial post-contingency rate of change of frequency (RoCoF) ($R^{\text{max}}$) and nadir frequency are considered by the coordinator to characterize the network's frequency response and the DER fleet's available capacity to respond.
 It might be necessary to lock the devices for a certain time after turning them ON for reliability issues or to avoid excessive switching. This can be done by setting $\eta_{\text{min}} > 0$. Using the designed controller, the constraint on $\eta_{\text{min}}$ leads to the following inequality:
\begin{align}
 \nonumber  \eta =  1-\frac{\Delta f_{\text{nadir}}-f_{\text{db}}}{f_{\text{max}}-f_{\text{db}}}-K_D{R}^{\text{max}}\ge \eta_{\text{min}}\\
    \Rightarrow K_D \le \frac{1}{{R}^{\text{max}} }\left(1-\frac{\Delta f_{\text{nadir}}-f_{\text{db}}}{f_{\text{max}}-f_{\text{db}}}-\eta_{\text{min}}\right). \label{eq:KDineq}
\end{align}
In addition to~\eqref{eq:KDineq}, the coordinator needs to design $K_P$ and $K_D$ with the corresponding minimum damping from~\eqref{eq:lowerbound} in mind relative to a desired predefined damping value.  

To illustrate these results, $\eta_{\text{min}}$ and minimum damping are plotted versus $K_D$ in Fig.~\ref{fig:etamin} for the simulation setup of Fig.~\ref{fig:different_pd}. For example, Fig.~\ref{fig:etamin} depicts the relationship between $\eta_\text{min}$, $K_\text{D}$ and  expected available damping. When $\eta_{\text{min}} \le 0.67$, then $K_\text{D} \le 8.30$, which means that the  expected available damping will be less than 4300 MW/Hz. The next section provides insight into practical considerations for packet-based DER coordination and primary frequency control via simulation-based analysis.
\begin{figure}
    \centering
    \includegraphics[width=1\linewidth]{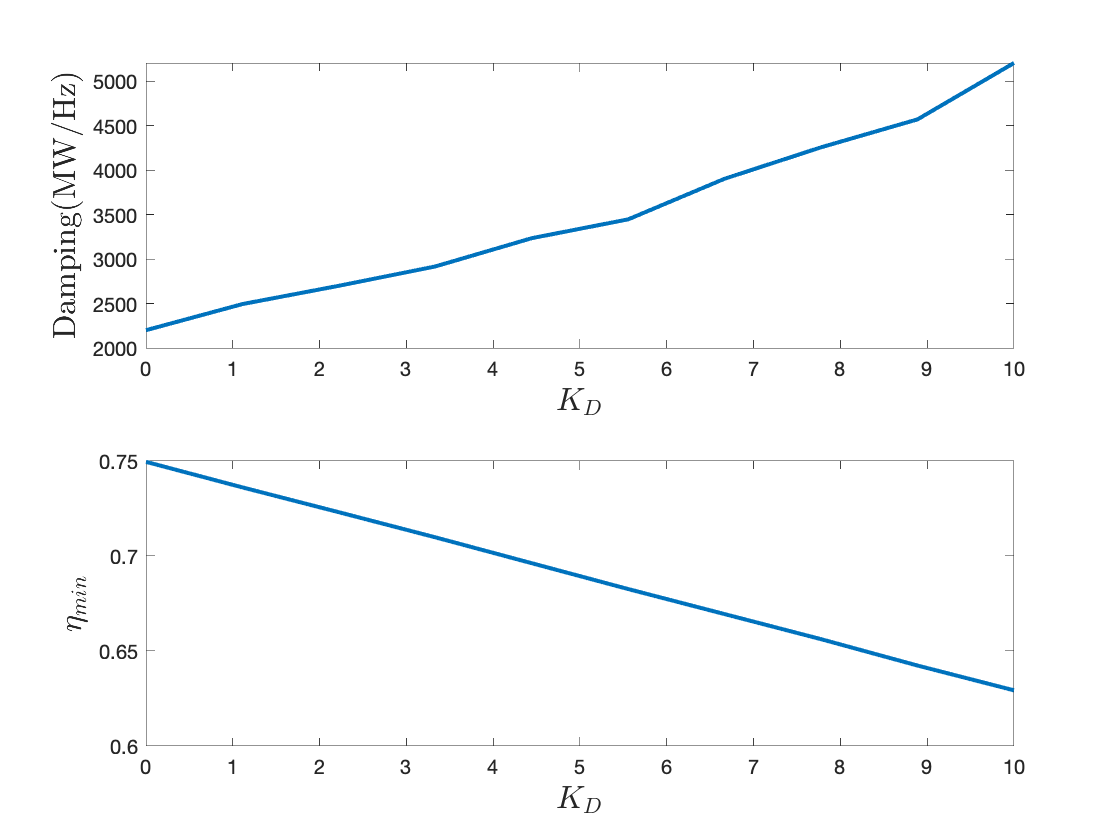}
    \caption{Unused timer and minimum damping for different values of $K_D$.}
    \label{fig:etamin}
\end{figure}

\section{Simulation and practical considerations}\label{sec:practical}
In this section, different practical considerations are tested to verify the performance of the proposed frequency-responsive controller for FFR. Different frequency measurement resolutions (between 1 mHz to 100 mHz) and random actuation delays (between 133 ms to 600 ms) are tested, and the effect on primary frequency control is analyzed. Moreover, simulation results are provided to determine how many DERs are needed to provide equal damping to an average droop-controlled generator in the network. The proposed decentralized controller is tested on the New England 39-bus system~\cite{moeini2015upec}. $K_D=2$ unless otherwise specified. A generation outage of 250 MW occurs at bus 30 at $t=2$ seconds. All of the DERs are connected to bus 20. The simulation setup is provided in table~\ref{t:setup}.

\begin{table}[htbp]
	\centering
	\caption{Simulation Parameters}	
	\begin{tabular}{ll} \toprule
	
		Parameter & Value\\ 
		\midrule
        Fleet size & 200,000 \\		
        $\Delta t$ & 10 ms  \\
        $(\Delta f_\text{db}, \Delta f_\text{max})$ & (36, 200) mHz  \\
		$T^\text{min/set/max}_{n}$& 48.8/52.0/55.2 $C^\circ $ \\ 
		($\delta$, MTTR) & (3,3) mins \\
        \bottomrule
	\end{tabular}
\label{t:setup}
\end{table}
\subsection{Actuation delay}
The effect of actuation delays is shown in Fig.~\ref{fig:delay}. With no actuation delay, the fleet responds to frequency deviation immediately after frequency deviation exceeds $f_{\text{db}}$. While the delay slightly affects the transient behavior, the impact on final frequency (and damping) is negligible. The results illustrate the acceptable performance of the proposed approach even with a 400 to 600 ms delay which is the typical delay value in practice.
\begin{figure}
    \centering
    \includegraphics[width=1\columnwidth]{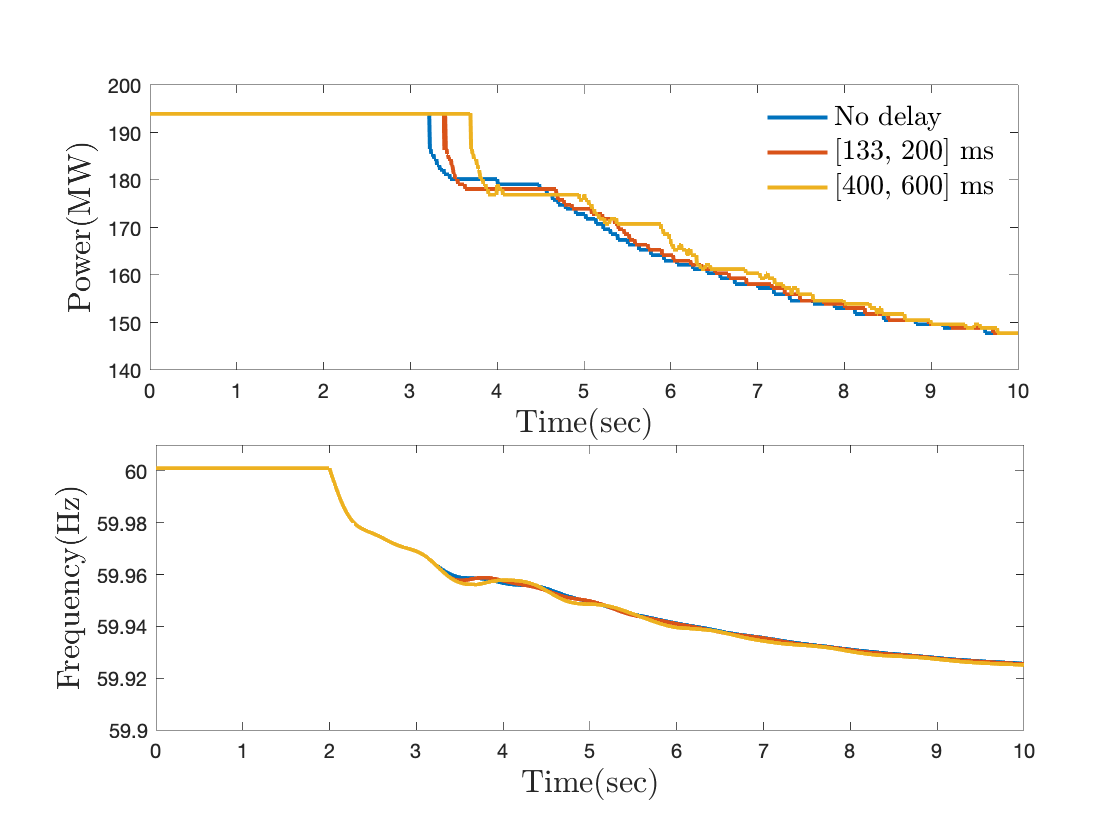}
    \caption{The frequency response for different actuation delays for a population of 200,000 DERs. The disturbance occurs at t=2 s. }
    \label{fig:delay}
\end{figure}
\subsection{Frequency measurement resolution}
Fig.~\ref{fig:resolution} shows the impact of frequency measurement resolution on frequency response. The root mean square error (RMSE) of power interruption is 0.84 MW for 1 mHz resolution, 4.25 MW for 10 mHz resolution, and 34.78 MW for 100 mHz resolution. It can be seen that for 10 mHz frequency measurement resolution, the results are close to the actual values. Therefore, the effect of measurement inaccuracy can be neglected if the measuring devices' accuracy is at least 10 mHz.   
\begin{figure}
    \centering    \includegraphics[width=1\columnwidth]{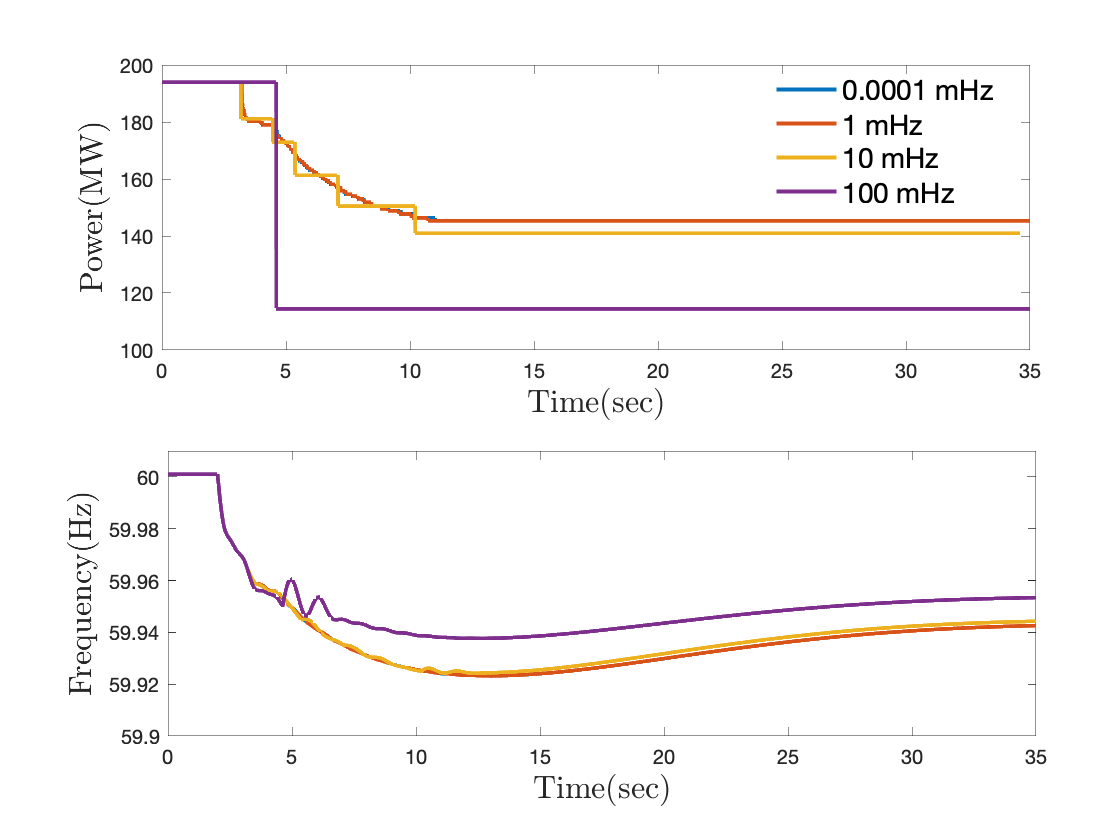}
    \caption{The frequency response for different frequency measurement resolutions for a population of 200,000 TCLs.}
    \label{fig:resolution}
\end{figure}

\subsection{Relating the scale of aggregate DER response}
When coordinating aggregations of DERs, it is of value to understand how many DERs are needed to replicate the synthetic available from a realistic power plant. To answer this question, an experiment is performed in which a generator in the IEEE 39-bus test system is tripped at bus~30, disconnecting 250 MW power. Then, the frequency at bus~39 is measured in two cases: $i)$ 1000 MW Generator at bus~39 with 5\% droop coefficient and no responsive loads $ii)$ deactivate the droop control at bus~39 and replacing it with 80,000 EWHs, each of which has a 4.5 kW power rating. The simulation results are presented in Fig.~\ref{fig:droop}. The yellow curve shows the frequency response at bus~39 without droop control and without DER coordination. The red curve shows the frequency response with 5\% droop control at bus~39, and the blue curve shows the frequency response when the droop controller is replaced with 80,000 EWHs with 4.5 kW power. The results show that 80,000 coordinated EWHs can provide a fast frequency response equivalent to a 1000 MW generation unit with a 5\% droop coefficient.   
\begin{figure}
    \centering
    \includegraphics[width=1\columnwidth]{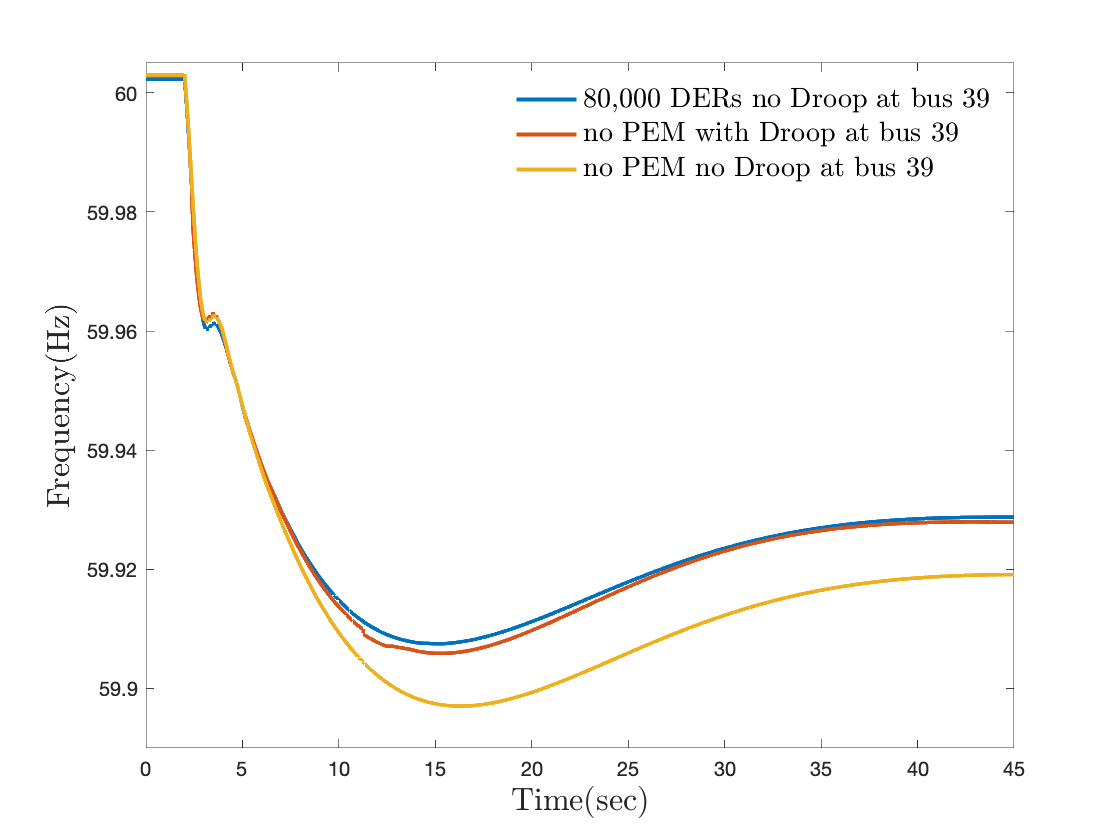}
    \caption{Frequency response, with/without droop control at bus 39 and with/without DER coordination.}
    \label{fig:droop}
\end{figure}
\section{conclusion}\label{sec:conclusion}
A fully decentralized frequency-based DER controller is designed and an analysis is presented that enables a DER aggregator to precisely estimate the synthetic damping available (online) from a fleet of aggregated DERs. To understand the impact of participating in ancillary services (i.e., frequency regulation) while also guaranteeing available synthetic damping from the fleet, a comprehensive analysis is provided to characterize a probabilistic lower bound on the available synthetic damping. This bound enables trade-off analysis between the fleet's ability to provide frequency regulation versus synthetic damping. Finally, practical considerations of the proposed decentralized control scheme are presented in a simulation-based study concerning the effects of actuation delays and frequency measurement resolutions. Future research directions include adapting the decentralized controller parameters based on spatial grid information to differentiate and prioritize certain locations/buses/feeders, as a way to incorporate the proposed synthetic damping with existing under-frequency load-shedding (UFLS) schemes. Another venue of interest represents the development of market mechanisms for incentivizing and valuing synthetic damping in low-inertia power systems.

\IEEEtriggeratref{22}
\bibliographystyle{IEEEtran}
\bibliography{ref2}

\end{document}